\documentclass[final,5p,times,English,twocolumn]{elsarticle}

\usepackage{hyperref}
\usepackage{caption}
\usepackage{indentfirst}

\usepackage{amssymb}
\usepackage{graphicx}
\usepackage{epstopdf}
\graphicspath{{fig/}}
\usepackage{subfigure}
\usepackage{float}
\usepackage{array}
\usepackage{url}
\usepackage{longtable}
\usepackage{rotating}
\usepackage{multirow}
\usepackage{amsfonts,amssymb}
\usepackage{mathrsfs}
\usepackage{mdwlist}
\usepackage{threeparttable}
\usepackage{makecell}
\usepackage{diagbox}
\usepackage{pifont}
\usepackage{appendix}

\usepackage{pdflscape}
\usepackage{footnote}

\usepackage{amsmath,amsthm}
\usepackage[linesnumbered,boxed,ruled]{algorithm2e}
\usepackage[hang,flushmargin]{footmisc}
\usepackage{lipsum}
\usepackage{fontawesome}

\usepackage{mdwlist}
\usepackage{booktabs}

\usepackage[T1]{fontenc}
\usepackage[utf8]{inputenc}

\usepackage{color,xcolor}
\definecolor{myred}{RGB}{217,46,127}
\definecolor{mygreen}{RGB}{67,127,127}

\newtheorem{thm}{Theorem~}
\newtheorem{lemma}{Lemma~}

\begin{document}

\begin{frontmatter}

\title{Eliminating Quantization Errors in Classification-Based Sound Source Localization}

\author[1,2]{Linfeng Feng}
\ead{fenglinfeng@mail.nwpu.edu.cn}

\author[1,2]{Xiao-Lei Zhang\corref{cor1}}
\ead{xiaolei.zhang@nwpu.edu.cn}

\author[3]{Xuelong Li}
\ead{li@nwpu.edu.cn}

\cortext[cor1]{Corresponding author.}

\affiliation[1]{organization={School of Marine Science and Technology, Northwestern Polytechnical University},
    city={Xi'an},
    citysep={}, 
    postcode={710072}, 
    state={Shaanxi},
    country={China}}
\affiliation[2]{organization={Research \& Development Institute of Northwestern Polytechnical University in Shenzhen},
    city={Shenzhen},
    citysep={}, 
    postcode={518063}, 
    state={Guangdong},
    country={China}}
\affiliation[3]{organization={Institute of Artificial Intelligence (TeleAI), China Telecom Corp Ltd},
    addressline={31 Jinrong Street}, 
    city={Beijing},
    citysep={}, 
    postcode={100033}, 
    country={P. R. China}}
\begin{abstract}
Sound Source Localization (SSL) involves estimating the Direction of Arrival (DOA) of sound sources. Since the DOA estimation output space is continuous, regression might be more suitable for DOA, offering higher precision. However, in practice, classification often outperforms regression, exhibiting greater robustness to interference. Conversely, classification's drawback is inherent quantization error. Within the classification paradigm, the DOA output space is discretized into intervals, each treated as a class. These classes show strong inter-class correlations, being inherently ordered, with higher similarity as intervals grow closer. Nevertheless, this has not been fully exploited. To address this, we propose an Unbiased Label Distribution (ULD) to eliminate quantization error in training targets. Furthermore, we tailor two loss functions for the soft label family: Negative Log Absolute Error (NLAE) and Mean Squared Error without activation (MSE(wo)). Finally, we introduce Weighted Adjacent Decoding (WAD) to overcome quantization error during model prediction decoding. Experimental results demonstrate our approach surpasses classification quantization limits, achieving state-of-the-art performance. Our code and supplementary materials are available at \href{https://github.com/linfeng-feng/ULD}{https://github.com/linfeng-feng/ULD}.
\end{abstract}


\end{frontmatter}

\section{Introduction} \label{sec:introduction}
Sound Source Localization (SSL) encompasses the task of determining the spatial coordinates of sound sources. Typically, this task is simplified to estimating the Direction of Arrival (DOA) of sound sources relative to microphones\cite{grumiaux2022survey}. The obtained DOA information holds substantial value, enhancing the performance of various downstream applications. One common example is sound event localization and detection (SELD) \cite{shimada2021accdoa,shimada2022multi,bai20233d}. Accurate DOA estimates facilitate effective multichannel speaker separation \cite{wang2022localization} and can serve as a basis for ordering labels of multiple speakers during training \cite{taherian2022lbt}. Integrating DOA as an input feature in Automatic Speech Recognition (ASR) models has demonstrated a notable reduction in word error rates \cite{subramanian2021directional, subramanian2022deep}. In complex acoustic environments, speaker diarization systems leveraging DOA information have exhibited substantial improvements \cite{zheng2021real, taherian2023multi}. However, real-world SSL encounters numerous challenges, including ambient noise, reverberation, and multiple speakers.

\subsection{Motivation and challenges}
Over the past few decades, most researchers have focused on developing SSL algorithms based on traditional array signal processing techniques \cite{knapp1976generalized, schmidt1986multiple, dibiase2000high}. In recent years, SSL research has shifted towards deep learning methods, where deep neural networks (DNNs) have demonstrated considerable promise and robustness in challenging acoustic environments, e.g. high reverberation \cite{chakrabarty2019multi}. Based on the training objective of DNN, deep learning-based DOA estimation can be typically categorized into two categories: regression and classification.

One of the early DNN-based regression methods \cite{vesperini2016neural} designs an output layer with two neurons, used to estimate the coordinate of a sound source in a 2D plane. A similar method in \cite{vecchiotti2019detection} also employs a similar output structure, but with an additional neuron for Voice Activity Detection (VAD). \cite{vera2018towards, adavanne2018sound} focuses on 3D Cartesian coordinate localization, where an additional neuron is used to estimate the source height.  \cite{shimada2021accdoa, shimada2022multi} introduces the concept of Activity-coupled Cartesian DOA (ACCDOA), intertwining sound activity with DOA to form labels for the trajectory regression task in SELD. Moreover, there are also regression methods that estimate the DOA of a source in the spherical coordinate system \cite{diaz2020robust, diaz2022direction}. 

One of the early classification-based approaches \cite{xiao2015learning} is implemented through a fully connected neural network. Note that classification-based approaches have natural advantages for multi-speaker localization. Assuming the speaker count is known, \cite{chakrabarty2019multi} uses phase spectra as the input of convolutional neural networks, and takes the multiple peak probabilities of the network output for multi-speaker localization. \cite{subramanian2022deep} proposed to decouple double-speaker localization into a set of single-speaker localization problems. Some work in the literature also studied scenarios where the source count is unknown. For example, \cite{he2019adaptation} compares the predicted distributions produced by a DNN with a threshold, where classes exceeding the threshold are considered to have source activity. \cite{nguyen2020robust} constructs a DNN with two output branches, one for outputting the distributions of sound sources, and the other for outputting the number of sources. \cite{fu2022iterative} presents an iterative SSL method that extracts the DOA of each sound source iteratively without a threshold.

Experiments by \cite{tang19_interspeech} show spherical DOA regression underperforms versus classification, while Cartesian regression outperforms classification. This echoes \cite{perotin2019regression} emphasizing greater precision for Cartesian regression and robustness for classification. The discrepancy results from significant quantization errors in Cartesian classification \cite{feng2023soft}.

Quantization error refers to the localization error when one-hot classification reaches 100\% accuracy. Quantization errors not only directly impact localization accuracy but also introduce non-smoothness in labels during training, potentially confusing the model. Consider a classification resolution of 5 degrees, with two samples having ground-truth DOA values of 87.6 and 92.4, resulting in identical one-hot labels representing 90. Since their DOA difference is 4.8, this leads to low intra-class similarity. Conversely, if the ground-truth DOA values are 92.4 and 92.6, their respective labels represent 90 and 95, with a DOA difference of 0.2, indicating high inter-class similarity.

Gaussian Label Coding (GLC) \cite{he2018deep} and Soft Label Distribution (SLD) \cite{subramanian2022deep} emerge as similar soft label strategies alternative to one-hot encoding. Both assign non-zero values to multiple classes near the ground-truth, smoothing labels. GLC allows highly smooth intra-and inter-class transitions without constraining label sum to 1. However, this flexibility prohibits the use of softmax activation in the output layer. 
SLD constrains label sum to 1, aligning with probabilistic interpretation. 
Theoretical advantage over one-hot encoding lies in inter-class smoothness, while quantization errors remain consistent. It is worth noting that these methods do not completely resolve the issue of label quantization errors during the training phase. Furthermore, during the decoding phase, they still select the only peak class from the output vector, reintroducing quantization errors.

\subsection{Goals and contributions}
Based on the aforementioned analysis, we propose a novel output architecture designed for classification. This architecture not only retains robustness of classification but also incorporates the high precision of regression. Across experiments in diverse environments, from ideal to challenging, we substantiate our proposed architecture's effectiveness. The primary contributions can be summarized as follows:
\begin{itemize}
    \item
    \textbf{We introduce an Unbiased Label Distribution (ULD) aimed at eliminating quantization errors during training.} ULD's injective encoding permits unbiased inverse mapping to ground truth position. Notably, ULD exhibits smooth transitions within and between classes. Furthermore, ULD retains advantages of one-hot encoding for classification.
    
    \item
    \textbf{We propose two novel loss functions for soft labels: Negative Log Absolute Error (NLAE) and Mean Squared Error without activations in the output layer (MSE(wo)).} Cross-entropy (CE) loss may be suboptimal for soft labels because its optimization objective does not directly point to the soft label itself. We analyzed compatibility between cross-entropy-like loss functions and classification model output layer activations, and advantages of MSE loss for soft labels. After analysis, our new strategy is to combine these loss types for the soft label family.

    \item
    \textbf{We propose a Weighted Adjacent Decoding (WAD) to address shortcomings of sole reliance on peak probability.} Selecting the class corresponding to the peak as the estimated DOA, denoted as {\textit{Top-1 decoding}}, suffers from quantization errors. We incorporate sidelobes of the peak class into decoding method design, yielding WAD. This overcomes the quantization error limit of Top-1 decoding.
\end{itemize}
The remainder of this paper is organized as follows. Section~\ref{sec:ssl} outlines the classification paradigm to introduce the issues. In Sections~\ref{sec:label} to \ref{sec:decoding}, we provide a detailed description of our contributions. Sections~\ref{sec:exp} and \ref{sec:res} demonstrate the effectiveness of our proposed method through experimental results. Finally, Section~\ref{sec:conclusion} presents the conclusions of our study.

\section{Supervised sound source localization} \label{sec:ssl}
In this paper, we focus solely on azimuth sound source localization. We begin with the single-source localization problem. The DOA is measured in degrees. Assuming the maximum output range of DOA is denoted as $r$, if microphones are collinear, then $r$ is 180; if microphones are coplanar but not collinear, then $r$ is 360. The classification model discretizes the output space of DOA into several cells, with the standard cell length denoted as $l$. To ensure boundary coverage, we set $I = r/l$ with $I \in \mathbb{N}$, yielding the set of class values $\{0, 1, \ldots, I-1, I\}$, so the output space of DOA is discretized into $\{0, l, \ldots, (I-1)\cdot l, I\cdot l\}$.

Without loss of generality, we assume that $u$ denotes an utterance. A DNN-based SSL model $\mathrm{DNN}(\cdot)$ can be formulated as:
\begin{equation}\label{eq:ssl_def}
    \begin{aligned}
    \boldsymbol{\kappa} = \mathrm{DNN}(u) \\
    \boldsymbol{\hat y} = \sigma(\boldsymbol{\kappa})
    \end{aligned}
\end{equation}
where the operation $\sigma(\cdot)$ maps $\boldsymbol{\kappa} \in \mathbb{R}^{I+1}$  to a predicted distribution $\boldsymbol{\hat y} \in [0, 1]^{I+1}$.

We assume that $\boldsymbol{y} \in [0, 1]^{I+1}$ represents the label distribution of the sound source's true position $p \in [0, r]$, and the process of obtaining $\boldsymbol{y}$ can be formalized as follows:
\begin{equation}\label{eq:label_encoding}
    \boldsymbol{y} = \mathrm{Encoding}(p)
\end{equation}
where $\mathrm{Encoding}(\cdot)$ represents the mapping from $p$ to $\boldsymbol{y}$. In general, if $\sum_{i=0}^I y_i=1$, then softmax activation is suitable as a candidate of $\sigma(\cdot)$ in Eq.~\eqref{eq:ssl_def}, otherwise sigmoid activation is more appropriate.

The training objective of a DNN is to find its optimal learnable parameters that minimize the loss function $\mathcal{L}$, while producing an output $\boldsymbol{\hat y}$ that is as close as possible to the ground-truth label distribution $\boldsymbol{y}$. This can be achieved through supervised training, which can be formulated as:
\begin{equation}\label{eq:deep_metric}
    \mathrm{DNN}^* = \arg \min_{\mathrm{DNN}} \mathcal{L}(\boldsymbol{y}, \boldsymbol{\hat y})
\end{equation}

After obtaining $\boldsymbol{\hat y}$ in the test stage, it becomes feasible to map $\boldsymbol{\hat y}$ to a corresponding predicted position $\hat p \in [0, r]$, which can be referred to as a decoding process:
\begin{equation}\label{eq:decoding}
    \hat p = \mathrm{Decoding}(\boldsymbol{\hat y})
\end{equation}
where $\mathrm{Decoding}(\cdot)$ represents the mapping from $\boldsymbol{\hat y}$ to $\hat p$.

\begin{figure}[t]
    \centering
    \includegraphics[width=0.45\textwidth]{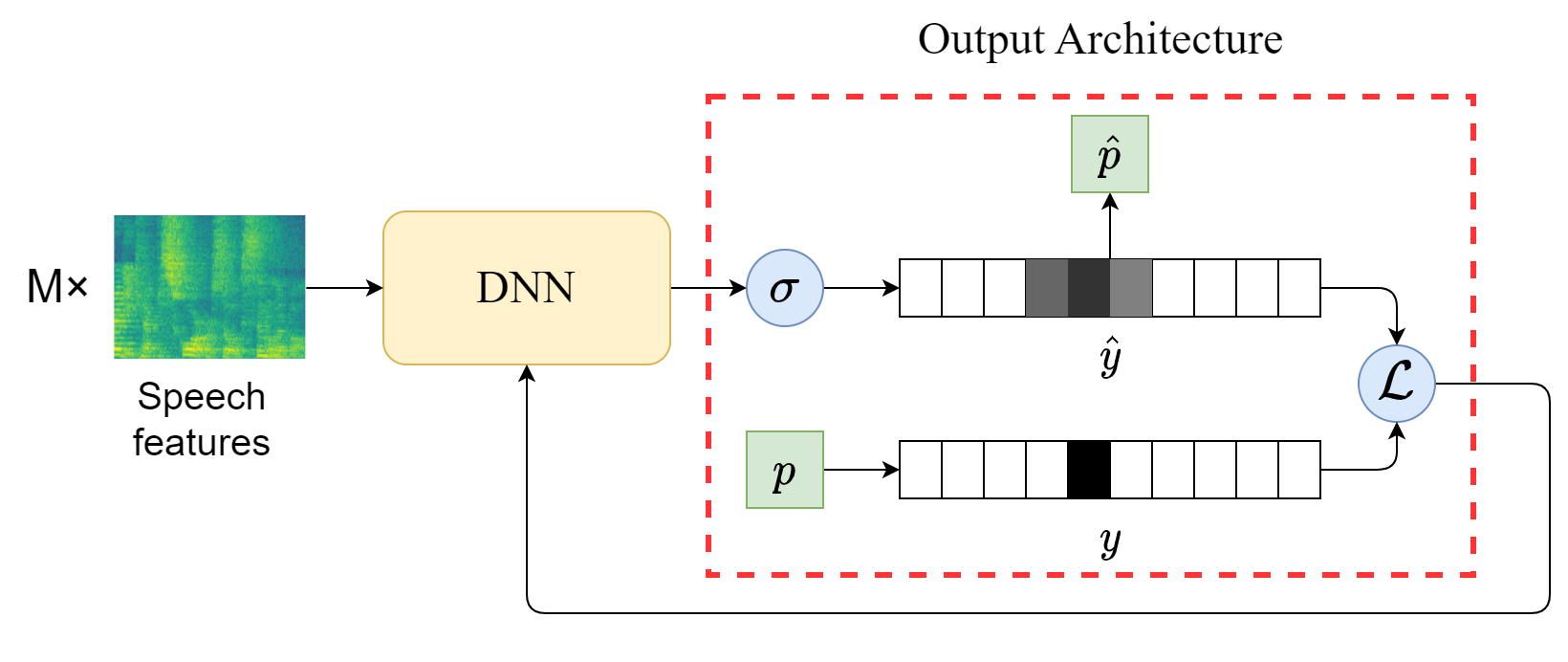}
    \caption{Workflow of a supervised sound source localization, where $M$ is the number of microphones. The output architecture is highlighted in the red box.}
    \label{fig:workflow}
\end{figure}

The workflow for supervised single-source localization can be represented by Figure \ref{fig:workflow}, where the red box represents the output architecture that is our main focus. The ultimate goal of our design is to improve the predictive accuracy of $\hat p$.

Concerning the challenge of localizing multiple sound sources, we suggest utilizing the source splitting mechanism proposed in \cite{subramanian2022deep} to decompose the problem into multiple single-source localization tasks. 

\section{Unbiased label distribution} \label{sec:label}
This section primarily concentrates on the label encoding presented in Eq.~\eqref{eq:label_encoding}, where we discuss the label distribution of a single speaker in a simple yet general manner.

\subsection{Analysis}\label{subsec:ana_ULD}
First, it should be emphasized that the true position of a sound source, $p$, is a real number, where $0 \leq p \leq r$. We introduce a scaled variable $\gamma = p / l$, where $\gamma \in \mathbb{R}$ and $0 \leq \gamma \leq I$. Typically, a common classification method is to apply an operation, $\mathrm{round}(\cdot)$, to assign $\gamma$ to its nearest integer and then encode it as a one-hot label distribution. From the perspective of probability, the scheme of assigning $p$ to the $\mathrm{round}(\gamma)$-th class can be interpreted as that, the probability of the sound source located in the $\mathrm{round}(\gamma)$-th cell is 1, while the probabilities in other cells are all 0. Formally, the one-hot label distribution is $\boldsymbol{y}^{\mathrm{1-hot}} = \{y^{\mathrm{1-hot}}_i\}_{i=0}^I$, with the code for the $i$-th class $y^{\mathrm{1-hot}}_{i}$ defined as:
\begin{equation}\label{eq:1-hot}
    y^{\mathrm{1-hot}}_{i} =
    \left\{\begin{array}{ll}
      1,& \mbox{ if   } i = \mathrm{round}(\gamma) \\
      0,& \mbox{ otherwise}
    \end{array}\right.,\quad \forall i = 0,\ldots, I
\end{equation}

From the above description, we can infer that the reason why the mapping from $p$ to the one-hot distribution encoding is not injective lies in the operation $\mathrm{round}(\cdot)$. In other words, it is inevitable to have quantization errors when using a single integer, $\mathrm{round}(\gamma) \in \mathbb{N}$, to represent a real number $p$.

\begin{thm}\label{thm:qe}
  The operation of $\mathrm{round}(\cdot)$ results in a quantization error whose mathematical expectation is $l / 4$.
\end{thm}

\begin{proof}
 See Appendix A for the proof.
\end{proof}

However, as $\sum_{i=0}^I y^{\mathrm{1-hot}}_{i}=1$, the one-hot distribution is consistent with probability theory interpretation, suitable for supervised models improving classification accuracy. Therefore, our motivation is finetuning this distribution addressing strengths and weaknesses.

\subsection{Definition}
\begin{thm}\label{thm:uld}
  Let $\gamma$ be a non-negative real number, with $\mathrm{int}(\gamma)$ representing its integer part, and $\mathrm{deci}(\gamma)$ denoting its decimal part. Then, for any $\gamma$ between two adjacent integers, $\mathrm{int}(\gamma)$ and $\mathrm{int}(\gamma)+1$, an unbiased approximation is given by:
\[
\gamma = (1 - \mathrm{deci}(\gamma)) \times \mathrm{int}(\gamma) + \mathrm{deci}(\gamma) \times (\mathrm{int}(\gamma) + 1)
\]
\end{thm}

\begin{proof}
 See Appendix B for the proof.
\end{proof}

Based on Theorem \ref{thm:uld}, we can easily derive our novel \textit{unbiased label distribution} $\boldsymbol{y}^{\mathrm{u}} = \{y^{\mathrm{u}}_i\}_{i=0}^I$ as follows:
\begin{equation}\label{eq:uld}
    y^{\mathrm{u}}_{i} =
    \left\{\begin{array}{ll}
      1-\mathrm{deci}(\gamma),& \mbox{ if   } i = \mathrm{int}(\gamma) \\
      \mathrm{deci}(\gamma),& \mbox{ if   } i = \mathrm{int}(\gamma)+1 \\
      0,& \mbox{ otherwise}
    \end{array}\right.,\quad \forall i = 0,\ldots, I
\end{equation}

Clearly, the sum of the elements of a ULD vector, represented as $\sum_{i=0}^I y^{\mathrm{u}}_{i}=1$, has a probabilistic interpretation, indicating the probability of a sound source appearing in the respective cell. The main advantage of ULD lies that:
\begin{thm}\label{thm:2}
ULD is free of quantization errors.
\end{thm}
\begin{proof}
It is clear that the encoding mapping from $p$ to $\boldsymbol{y}^{\mathrm{u}}$ is a one-to-one mapping, such that when $p_1 \neq p_2$, we have $\boldsymbol{y}^{\mathrm{u}}_1 \neq \boldsymbol{y}^{\mathrm{u}}_2$, enabling $\boldsymbol{y}^{\mathrm{u}}$ to be accurately inverse-mapped to $p$.
\end{proof}

\subsection{Connection between ULD and one-hot distribution}
When $\mathrm{deci}(\gamma)<0.5$, $\mathrm{int}(\gamma)=\mathrm{round}(\gamma)$, otherwise $\mathrm{int}(\gamma)+1=\mathrm{round}(\gamma)$. Therefore, ULD can be regarded as a smoothed one-hot distribution, which can fully inherit the advantages of the one-hot distribution, while avoiding the problem of disproportionate distribution similarity to the DOA distance.


\begin{figure}[t]
    \centering
    \includegraphics[width=0.45\textwidth]{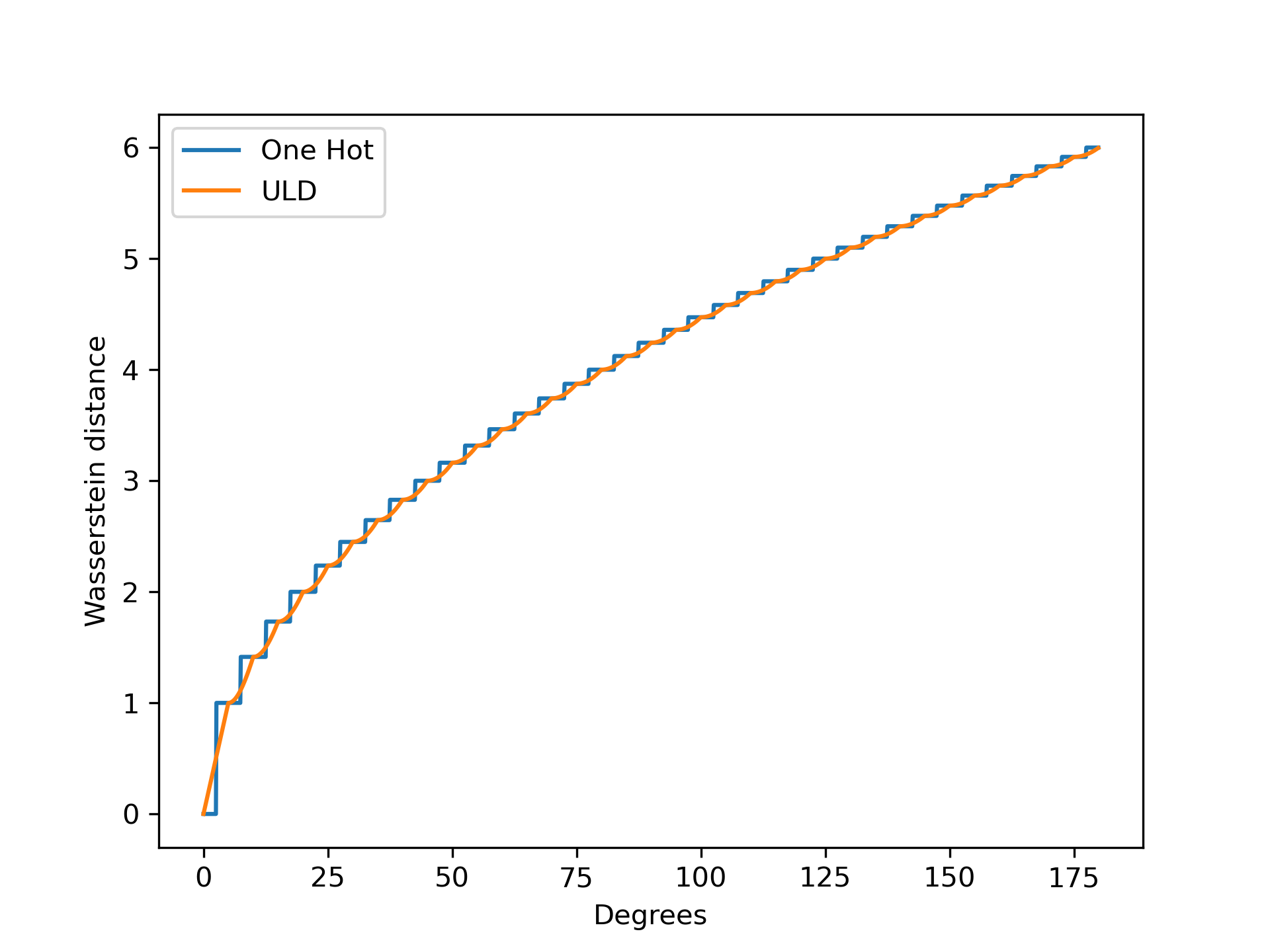}
    \caption{Wasserstein distance between a distribution of any angle between $[0, 180]$ degrees and the distribution of 0 degree, where the label distribution is one-hot or ULD.}
    \label{fig:distributions_wd}
\end{figure}

We use Wasserstein distance (WD) \cite{rubner2000earth} to analyze the connection between ULD and one-hot distribution. WD is a metric that can be used to measure the distance between two discrete probability distributions. 
Figure~\ref{fig:distributions_wd} shows the WD between the arbitrary-angle label distribution and 0-angle distribution of either ULD or one-hot. We see that the WD curves of the ULD and one-hot are closely related. However, unlike the sudden changes in the WD curve of the one-hot distribution, the WD curve of ULD is smooth.

\section{Loss functions}\label{sec:loss}
Once an appropriate label distribution has been established, our next task is to design a reasonable loss function $\mathcal{L}$ in Eq.~\eqref{eq:deep_metric} to train the DNN for yielding a predicted distribution closely matching the label distribution.

\subsection{Analysis}\label{subsubsec:anal2}
We believe that a reasonable $\mathcal{L}$ should have two key properties: (i) the direction of backpropagated gradient should always be correct, and (ii) the magnitude of backpropagated gradient should be proportional to the deviation from the training target to the DNN's output $\boldsymbol{\kappa}$. However, none of the common loss functions, including CE, BCE, and MSE, satisfy both properties for soft labels, as will be analyzed in Section \ref{subsubsec:anal2}. This analysis motivates deriving the NLAE loss and MSE (wo) loss in Section \ref{subsubsec:NLAE}, satisfying both properties simultaneously.

\subsubsection{ CE loss function}
The CE loss function is commonly used in conventional multi-classification problems:
\begin{equation}\label{eq:ce}
    \mathcal{L}^{\mathrm{CE}}=-\sum_{i=0}^I y_i \log \hat{y}_i
\end{equation}
where CE does not directly impose penalties on the classes with zero values in $\boldsymbol{y}$.

\subsubsection{BCE loss function}

In contrast to CE, the BCE loss function has a global receptive field:
\begin{equation}\label{eq:bce}
    \mathcal{L}^{\mathrm{BCE}}=-\sum_{i=0}^I y_i \log \hat{y}_i + (1-y_i) \log (1-\hat{y}_i)
\end{equation}

\begin{lemma}\label{lemma:bce}
  Given a variable set $\{\hat y_1,\hat y_2, ...,\hat y_I\}$, subject to the constraint $\sum_{i=1}^I \hat y_i = c$, where $c$ is a constant real number in $[0, 1]$, the minimum value of $-\sum_{i=1}^I \log(1- \hat y_i)$ is attained when all elements within the set are equal.
\end{lemma}
\begin{proof}
 See Appendix C for the proof.
\end{proof}

The primary distinction between BCE and CE lies in the divergent losses arising from their application to incorrect classes. Without loss of generality, consider a scenario with $I$ zeros in the label. Substituting the label and predicted distributions into BCE yields a loss of $-\sum_{i=1}^I \log(1- \hat y_i)$ for this portion. From Lemma \ref{lemma:bce}, we can infer that this loss is minimized when the incorrect classes in the predicted distribution assume equal values. In conventional multi-classification, classes are typically treated as unrelated. Therefore, the probability values for incorrect classes in the predicted distribution usually similar, aligning subtly with Lemma \ref{lemma:bce}. Consequently, CE optimality emerges, or alternatively, the BCE loss for incorrect classes appears almost indistinguishable. 

However, CE formulation becomes suboptimal for SSL classification. In SSL, class similarity is exceedingly high, prompting DNN output distributions to manifest undesired sidelobes around ground truth classes and even yielding \textit{pseudo peaks}. Given the reverberation, the likelihood of pseudo peaks occurring will escalate. However, these pseudo peaks are not directly perceptible by CE, as depicted in Eq.~\eqref{eq:ce}. Specifically, the highly non-linear amplification of values by the negative log function (as $\hat y_i$ approaches 1, $-\log(1-\hat y_i)$ approaches infinity) results in substantial loss within BCE's second portion when pseudo peaks assume elevated values.
\begin{figure*}[t]
  \centering
  \subfigure[Distribution A]
  {
  \includegraphics[width=0.23\textwidth]{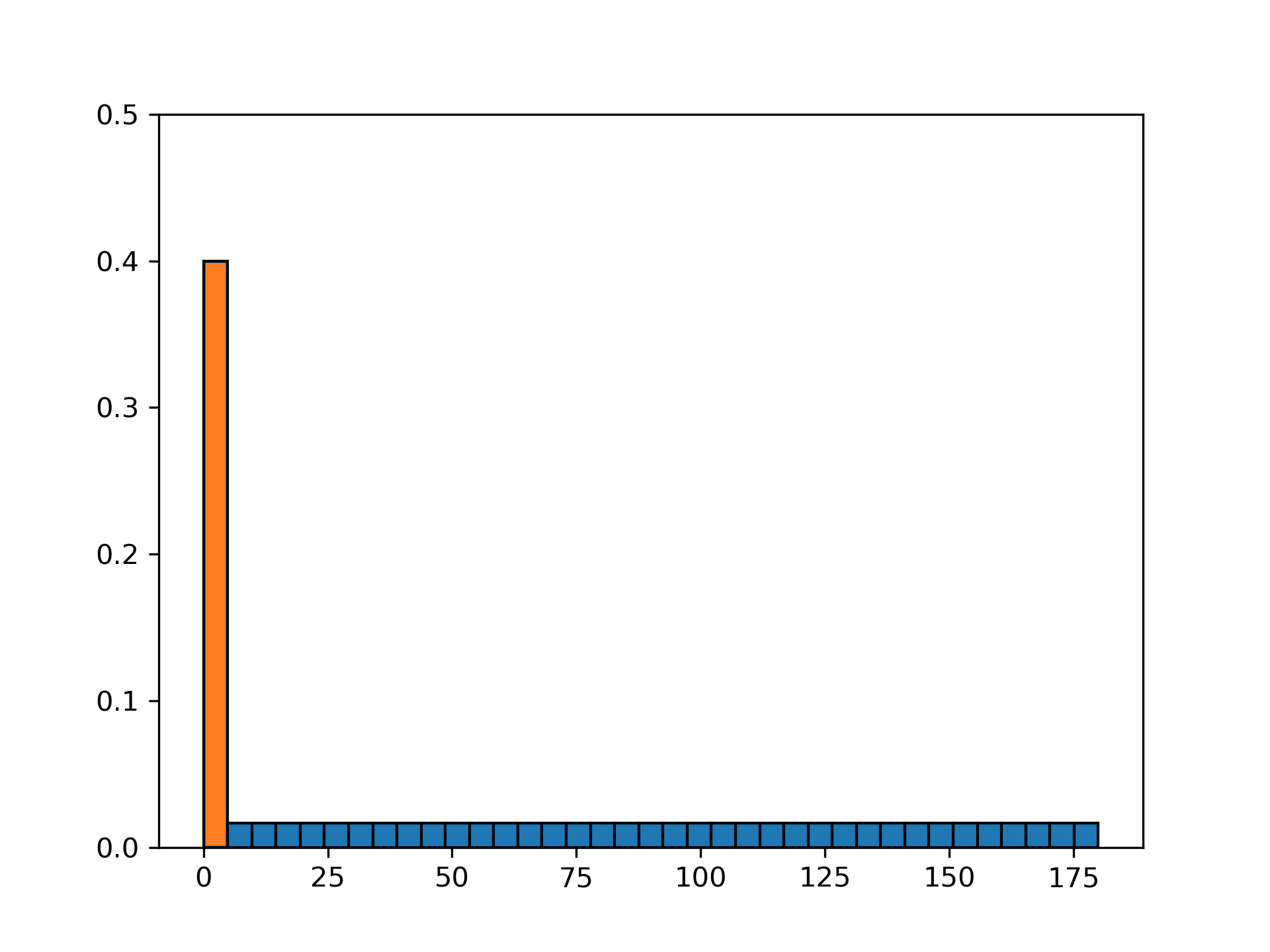}
  \label{fig:prediction_a}
  }
  \subfigure[Distribution B]
  {
  \includegraphics[width=0.23\textwidth]{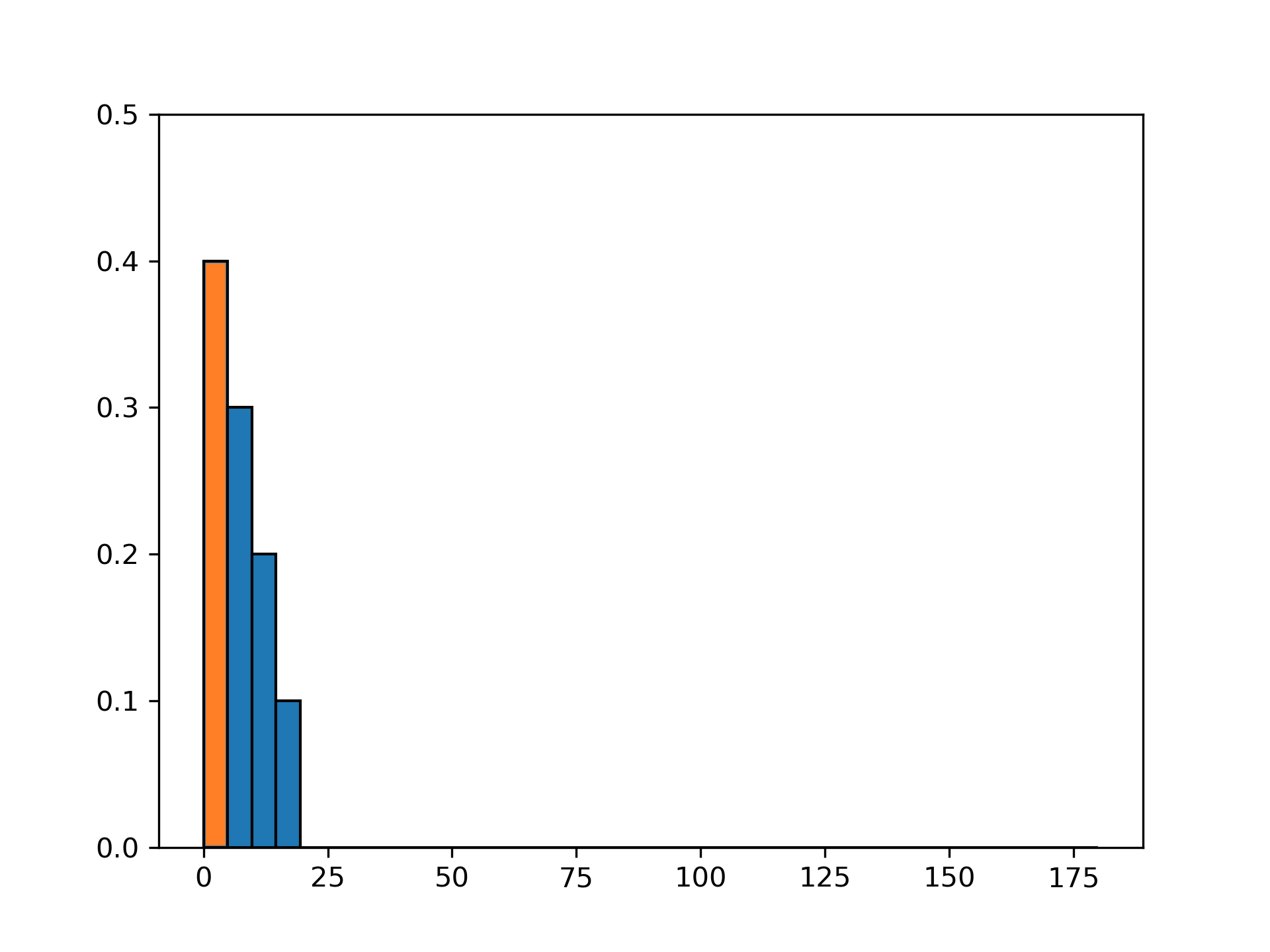}
  \label{fig:prediction_b}
  }\subfigure[Distribution C]
  {
  \includegraphics[width=0.23\textwidth]{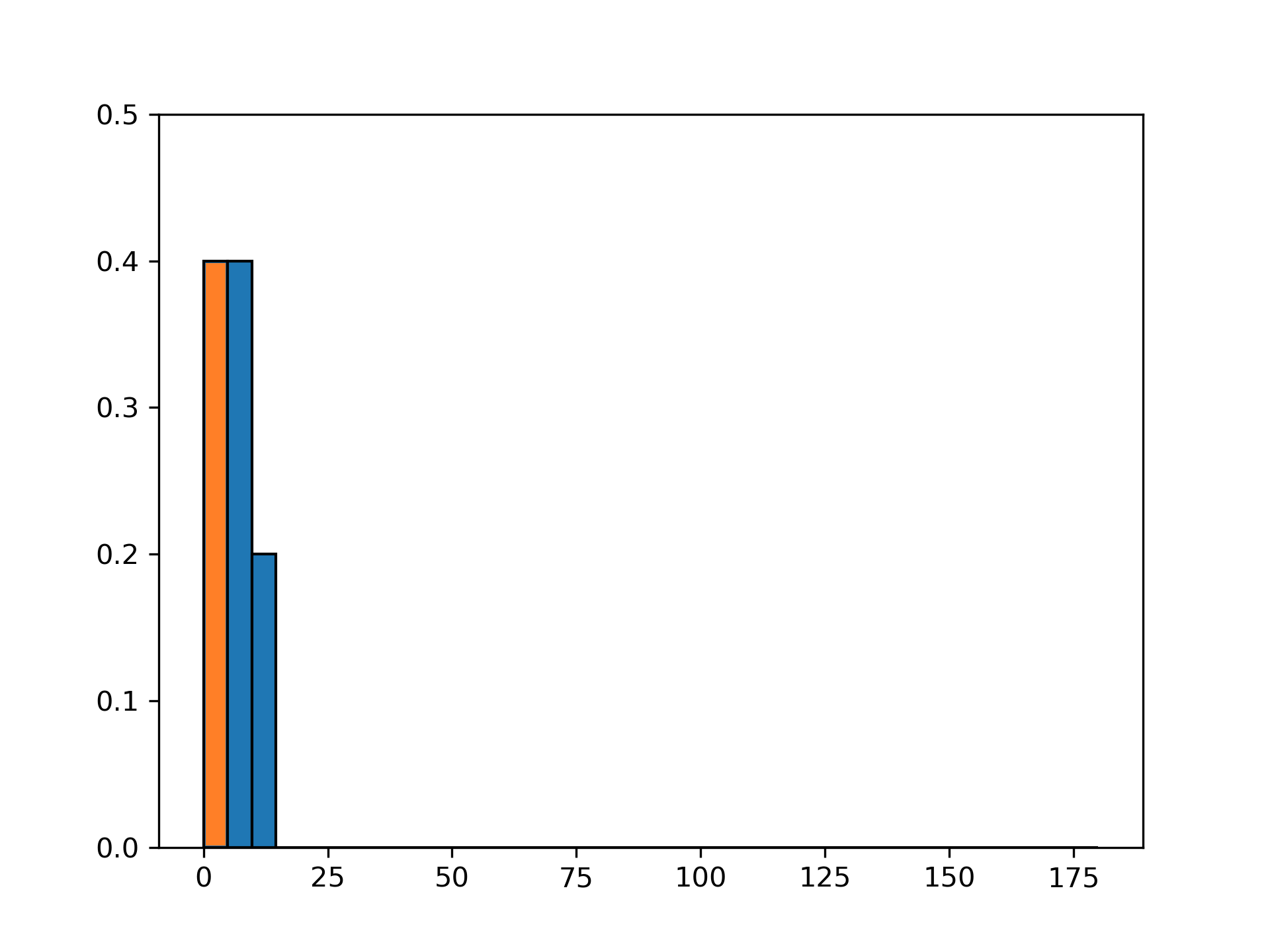}
  \label{fig:prediction_c}
  }\subfigure[Distribution D]
  {
  \includegraphics[width=0.23\textwidth]{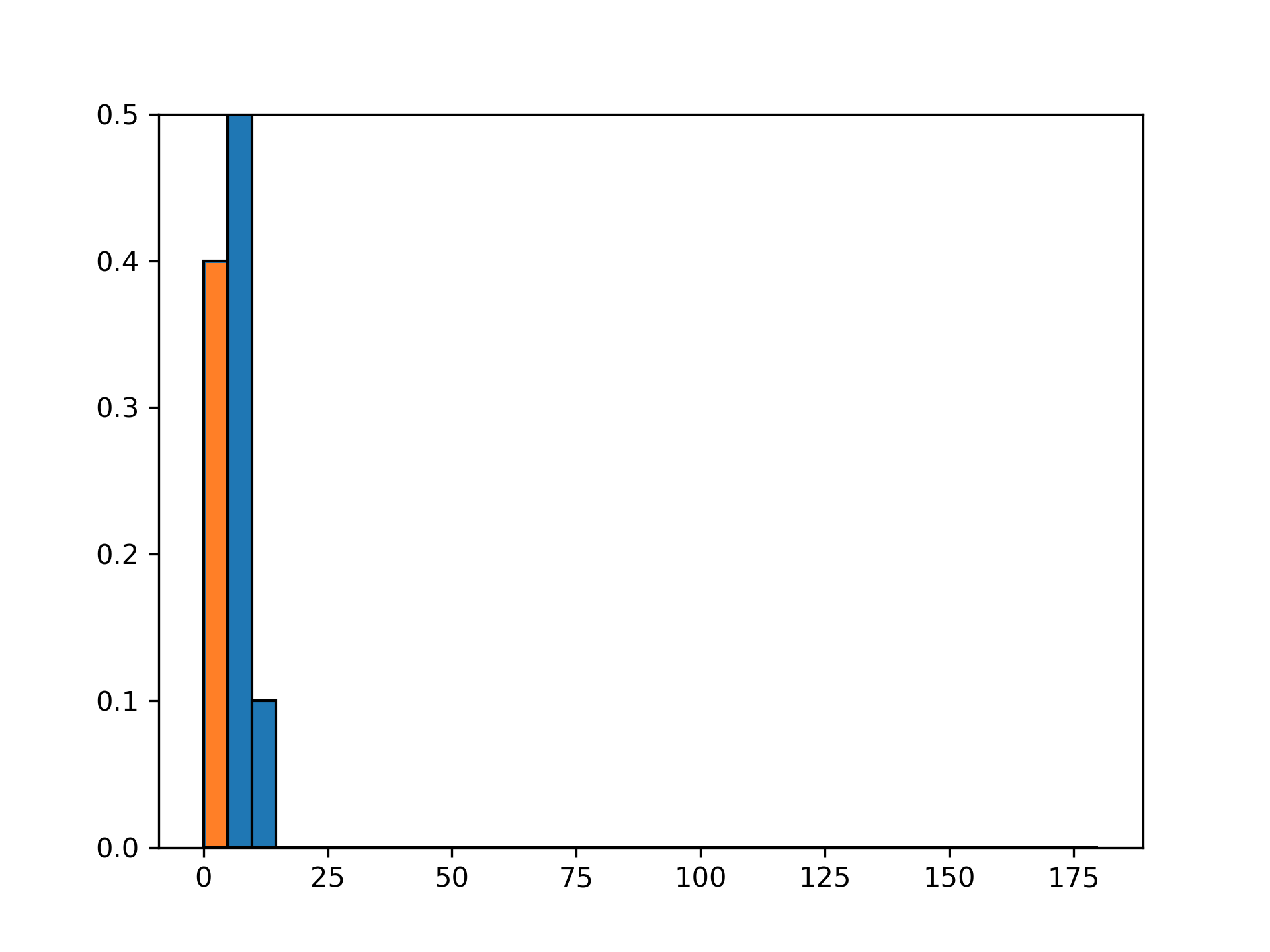}
  \label{fig:prediction_d}
  }
  \caption{An example on the advantage of BCE over CE for SSL. Consider a sound source with $r$ is 180, $l$ is 5 and $p$ is 0. In conventional classification problems, it is easy to have a predicted distribution like distribution A. However, for SSL, it is easy to have predicted distributions like the latter three. Distribution B has unwanted sidelobes, while distributions C and D even have pseudo peaks. The occurrence of pseudo peaks means that classification errors have already occurred. These four distributions have equal probability values in the ground truth class, so the CE loss of these distributions is equal. However, the BCE losses of these four are 1.52, 1.60, 1.65 and 1.71 respectively, meaning that BCE gives greater punishment to these negative factors that are easily encountered in SSL.}
  \label{fig:prediction}
\end{figure*}

Figure~\ref{fig:prediction} provides a concrete example to visually convey our theoretical analysis. The SSL problem is more likely to occur in distributions that significantly deviate from the conditions in Lemma \ref{lemma:bce}. Therefore, we recommend using a loss function with a global receptive field, such as BCE, to suppress pseudo peaks.

However, BCE exhibits an additional defect when dealing with soft labels. Focusing on the non-zero value classes of soft labels, we find BCE simultaneously propagates gradients in two opposite directions (the first part points $\hat y_i$ to 1, the second part points $\hat y_i$ to 0, not $y_i$). In other words, with soft labels, a cross-entropy-like loss has a minimum loss that is not zero, which may be suboptimal.

\subsubsection{MSE loss function}
In contrast to cross-entropy-like loss functions, an alternative approach can be explored by utilizing a loss function that aims for an ideal loss of zero. The MSE loss function, which also possesses a global receptive field, is defined as follows:
\begin{equation}\label{eq:mse}
    \mathcal{L}^{\mathrm{MSE}}=\sum_{i=0}^I (y_i - \hat{y}_i)^2
\end{equation}

It is self-evident that the gradient direction passed back through the MSE function ensures that $\hat{y}_i$ consistently points towards $y_i$. In Eq.~\eqref{eq:ssl_def}, the function $\sigma(\cdot)$ is a highly nonlinear transformation when ${\kappa}_i$ is either very large or very small. As a result, the MSE function often suffers from the problem of gradient disappearance when $\hat{y}_i$ approaches 0 or 1, and is therefore not frequently used for optimizing classification models.

\subsection{Definition}\label{subsubsec:NLAE}

\subsubsection{NLAE loss function}
Given the aforementioned analyses, we have devised a loss function called Negative Log Absolute Error (NLAE). It can be defined as follows:
\begin{equation}\label{eq:nlae}
    \mathcal{L}^{\mathrm{NLAE}}=-\sum_{i=0}^I \log(1 - |y_i - \hat{y}_i|)
\end{equation}

It is evident that the optimization direction of NLAE is to make $1 - |y_i - \hat{y}_i|$ approach 1, which is equivalent to make $y_i = \hat{y}_i$. In particular, when $y_i$ is 0 or 1, $\mathcal{L}^{\mathrm{NLAE}}$ and $\mathcal{L}^{\mathrm{BCE}}$ are equivalent, so the magnitude of the gradient passed back through NLAE is reasonable. Hence, the NLAE loss function is theoretically more suitable for the family of soft labels.

\subsubsection{MSE(wo) loss function}
Alternatively, we can approach the problem from a different perspective. Since the nonlinearity is caused by $\sigma(\cdot)$, we can consider discarding it. In other words, the MSE loss function can be modified to directly operate on $\boldsymbol{\kappa}$, formulated as:
\begin{equation}\label{eq:msewo}
    \mathcal{L}^{\mathrm{MSE (wo)}}=\sum_{i=0}^I (y_i - \kappa_i)^2
\end{equation}
where $\mathcal{L}^{\mathrm{MSE (wo)}}$ represents the MSE loss function without $\sigma(\cdot)$. However, if we take this approach, we need an additional operation to ensure that the predicted distribution does not exceed the boundary values, which can be expressed as $\boldsymbol{\hat y} = \{\mathrm{min}(\mathrm{max}(0, \kappa_i), 1)\}_{i=0}^I$.

The above discussion pertains to loss functions for single label. If necessary, adopting the strategy mentioned in  \cite{feng2023soft}, which calculates the loss for multiple labels separately and then adds them with weights to jointly optimize the DNN, may improve performance.

\section{Weighted adjacent decoding}\label{sec:decoding}
After training a DNN using ULD and the proposed loss functions, we obtain a predicted distribution $\boldsymbol{\hat y}$. This section describes our new approach for transforming $\boldsymbol{\hat y}$ into a DOA estimation $\hat p$.

\subsection{Analysis}
Naturally, we first extract a class $\hat{k}$ corresponding to the peak probability in $\boldsymbol{\hat y}$, which can be represented as:
\begin{equation}\label{eq:peakclass}
    \hat k = \arg \max_i \{\hat{y}_i\}_{i=0}^I
\end{equation}
Then, the source location can be obtained from $\hat{k}$ as:
\begin{equation}\label{eq:top1}
    \hat p = \hat k \cdot l
\end{equation}
which we refer to as the Top-1 Decoding. However, as discussed in Section~\ref{subsec:ana_ULD}, this formulation inevitably introduces quantization error. Even if a DNN achieves 100\% classification accuracy, the mathematical expectation of the quantization error using this decoding is $l/4$.

\subsection{Definition}
For predicted distributions produced from DNN, classes adjacent to the peak often have non-negligible probabilities due to strong correlation and ordering in the discretized DOA output space. Leveraging this, we propose Weighted Adjacent Decoding (WAD). We first pad $\boldsymbol{\hat y}$ with zeros (this step can also be interpreted as setting the probability of the sound source being outside the output space to 0). Then, for specified $\hat k_l$ and $\hat k_r$ classes to left and right of $\hat k$, respectively.

Initially, we consider the scenario where only the peak class $\hat k$ and the adjacent class $\hat k_h$ with relatively high probability are used. Specifically, $\hat k_h$ is defined as $\arg\max_i \{\hat{y}_i\}_{i=\{\hat k_l, \hat k_r\}}$. Since two classes are involved, we refer to this as WAD-2. Hence, WAD-2 is formalized as:
\begin{equation}\label{eq:wad2}
    \hat p = \frac{\sum_{i=\{\hat k,\hat k_h\}} \hat y_i \times i \times  l}{\sum_{i=\{\hat k,\hat k_h\}} \hat y_i}
\end{equation}

In practice, the probability values of the classes represented by $\hat k_l$ and $\hat k_r$ are usually quite high. Thus, it is also intuitive to consider a similar approach, WAD-3, which can be formalized as follows:
\begin{equation}\label{eq:wad3}
    \hat p = \frac{\sum_{i=\{\hat k, \hat k_l, \hat k_r\}} \hat y_i \times i \times  l}{\sum_{i=\{\hat k, \hat k_l, \hat k_r\}} \hat y_i}
\end{equation}

\subsection{Connection between WAD and Top-1 decoding}
From the above formulation, we observe that Top-1 decoding defined in Eq.~\ref{eq:top1} is a special case of WAD with $i=\{\hat{k}\}$. In other words, WAD extends Top-1 decoding by using a weighted combination of multiple classes to mitigate quantization error. In practice, employing more classes than WAD-3 may yield less robust decoding performance. Substituting the ULD from Eq.~\ref{eq:uld} into Eq.~\ref{eq:top1}, the quantization error remains the same as one-hot. However, substituting the ULD into Eq.~\ref{eq:wad2} or Eq.~\ref{eq:wad3} yields zero quantization error. Hence, this integrated output architecture is self-consistent.

\section{Experimental setup}    \label{sec:exp}
\subsection{Datasets}

\begin{table*}[t]
    \centering
    \caption{Specifications of the simulated data.}
      \scalebox{0.92}{\begin{tabular}{ccccccc}
      \toprule
      \multicolumn{2}{c}{Dataset} & C1    & C2    & A1    & L1    & L2 \\
      \midrule
      \multicolumn{2}{c}{Shape of array (m)} & \multicolumn{2}{c|}{Circular, radius = 0.05} & \multicolumn{3}{c}{Linear, aperture = 0.08} \\
  \cmidrule{3-7}    \multicolumn{2}{c}{Self-rotation angle of array (degree)} & 0     & 0     & 0     & [0, 180] & [0, 180] \\
      \multicolumn{2}{c}{Distance from speaker to array (m)} & 1.5   & 1.5   & 1.5   & [0, 14.1] & [0, 14.1] \\
      \multicolumn{2}{c}{Minimum distance from speaker to wall (m)} & 0.5   & 0.5   & 0.5   & 0.0   & 0.0 \\
      \multicolumn{2}{c}{Number of sound sources} & 1     & 2     & 1     & 1     & 2 \\
      \multirow{3}[0]{*}{Reverberation (s)} & train & [0.2, 0.7] & [0.2, 0.7] & anechoic & [0.2, 1.2] & [0.2, 1.2] \\
            & validation & [0.2, 0.7] & [0.2, 0.7] & anechoic & [0.2, 1.2] & [0.2, 1.2] \\
            & test  & [0.2, 0.8] & [0.2, 0.8] & anechoic & [0.2, 1.2] & [0.2, 1.2] \\
      \multirow{3}[1]{*}{Segments} & train & 36000 & 36000 & 18000 & 36000 & 72000 \\
            & validation & 3600  & 3600  & 1800  & 3600  & 7200 \\
            & test  & 4320  & 4320  & 1800  & 3600  & 7200 \\
      \bottomrule
      \end{tabular}}
    \label{tab:simu_dataset}%
  \end{table*}%

In this section, we conducted experiments on both simulated and real-world data. All source speech came from the LibriSpeech corpus \cite{panayotov2015librispeech}. The train-clean-360, dev-clean and test-clean subsets were used to generate corresponding subsets of simulated datasets. We utilized the Pyroomacoustics \cite{scheibler2018pyroomacoustics} module to generate room impulse responses. For each utterance, we randomly set a room size and selected a 2-second segment to generate multi-channel data with reverberation. Each multi-source speech signal was mixed from different speakers. Additionally, we introduced additive noise to the reverberant speech. The additive noise was randomly selected from a large-scale noise set \cite{tan2021speech} containing 126 hours of various types of noises. The signal-to-noise ratio (SNR) of each utterance was randomly chosen from a range of $[10, 20]$ dB. The training, validation, and testing sets had non-overlapping subsets of additive noise, ensuring independence between them.

As shown in Table~\ref{tab:simu_dataset}, we created five sets of simulated datasets, denoted as C1, C2, A1, L1, and L2 respectively, with different acoustic conditions to test the effectiveness and reliability of the proposed method. All datasets use microphone arrays consisting of 4 microphones. The complexity of the acoustic environment is mainly reflected in the reverberation and far-field, covering a wide range from anechoic to highly reverberant. The impact of number of sound sources is also considered in our datasets.

The length and width of a room were randomly chosen within $[4, 10]$m, the height of the room was fixed at 3.2 m, and the height of the sound source and microphone was fixed at 1.3 m. The reverberation time T$_{60}$ of the room was randomly selected within a given range, or the room is set to be anechoic.


For the C1, C2, and A1 datasets, each random room produces a single utterance. For the L1 and L2 datasets, a random room plays an audio once, but with 10 microphone arrays in the room to capture signals. Consequently, each room in L1 and L2 generates 10 distinct DOA segments. The placement of both microphones and sound sources is randomized. As a result, L1 and L2 are two datasets without any constraints on the distance between sound sources and microphone arrays or between sound sources and walls.

We recorded a real-world dataset \cite{liu2022deep} in two scenarios: an office and a conference room respectively. The office room is approximately $10.3\times9.8\times4.2$m\ with a T$_{60}$ of approximately 1.39s. The conference room is approximately $4.26 \times 5.16 \times 3.16$m with a T$_{60}$ of approximately 1.06s. The additive noise in both rooms can be ignored. We used the test-clean subset of LibriSpeech as the sound source to play back in the room, with different speakers corresponding to sound sources played at different locations. The equipment used to record the dataset was one speaker and 10 linear arrays, each with the same shape as those used in L1 and L2. After being divided into 2-second segments, each room had a total of 97,480 samples. We randomly selected segments to generate a subset with two speakers. Therefore, each room contains a total of 7,200 multi-speaker samples. We used the real-world data only for testing, while the simulated datasets of L1 and L2 were used for training and selecting models.

\subsection{Comparison among different label encoding}
We compare ULD with three label encoding, which are:
\begin{itemize}
  \item \textbf{One-hot}.
  \item \textbf{Gaussian Label Coding} \cite{he2018deep}: We followed \cite{he2018deep} and set the standard deviation to 8
  \item \textbf{Soft Label Distribution} \cite{subramanian2022deep}.
\end{itemize}

\subsection{Comparison among different loss functions}
We compare NLAE and MSE (wo) with four loss functions, which are:
\begin{itemize}
  \item \textbf{Cross Entropy (CE)}.
  \item \textbf{Binary Cross Entropy (BCE)}.
  \item \textbf{Mean Squared Error (MSE)}.
  \item \textbf{Wasserstein Distance (WD)} \cite{subramanian2022deep}: \cite{levina2001earth} proved that under specific conditions, WD can be simplified to Mallows distance with a closed-form solution. Based on this, \cite{subramanian2022deep} used it to optimize the SSL models.
\end{itemize}

Note that both CE and WD only apply to label with a sum of 1, so they are not suitable for GLC.

\subsection{Neural networks}
Four neural networks served as backbone networks, and the specific architectures are detailed in the supplementary material. The first network is the Phase Neural Network (PNN) \cite{chakrabarty2019multi}, comprising three convolutional layers and three dense layers. The second network is PNN-Split, a modified version of PNN. Notably, PNN-Split routes the output of the first dense layer through a recurrent layer for implicit speech separation \cite{subramanian2022deep}. Finally, the separated features are passed through another dense layer. The third network is SNet \cite{he2021neural}, while the fourth network is a hybrid model combining PNN-Split and SNet, known as SNet-Split. SNet-Split adopts all feature extraction modules of SNet, flattens the embedding features from the last residual block, and follows the subsequent operations consistent with PNN-Split.

\subsection{Inputs of networks}
We used a sampling rate of 16 kHz, a window length of 512 samples, a hop length of 256 samples, a Hanning window, and 512 FFT points to extract Short-Term Fourier Transform (STFT) features. For PNN, the input is a single frame of the phase spectrum, and the final output probability distribution is averaged along the time dimension during post-processing to obtain the final utterance-level localization probability distribution. For SNet, the input is 7 consecutive frames of STFT, with the real and imaginary parts of the STFT concatenated along the microphone channel dimension. Post-processing also involves averaging along the time dimension. For the networks with Split, the recurrent layer learns two masks. These masks is first multiplied by the embedding features, and then averaged along the time dimension to locate each sound source, therefore, no post-processing is required.

\subsection{Training and evaluation details}
For all experiments, we employed the AdamW optimizer with a batch size of 32 and a maximum of 30 training epochs. The learning rate was initialized at 0.001 and reduced to 0.0001 if the validation loss did not improve over 3 consecutive epochs. Training was terminated early if the model’s loss on the validation set did not improve for 10 consecutive epochs. We selected the model with the minimum localization error on the validation set for evaluation. The PNN and SNet were trained and tested on single-source datasets, while the Split networks was used for multi-source datasets. Since the DOA space is inherently ordered, it is easy to train multi-source models using location-based training \cite{taherian2022lbt}. For the dataset of C1, $l$ was set to 3; for C2, $l$ was set to 8; for A1, L1, and R1, $l$ was set to 5; and for L2 and R2, $l$ was set to 7.5.

\subsection{Evaluation metrics}
Suppose a dataset has $N$ test speakers. As we primarily discusses classification models, a natural evaluation metric is classification accuracy (ACC), which can be formalized as follows:
\begin{equation}
    \mathrm{ACC}(\%)=\frac{N^\mathrm{acc}}{N} \times 100
\end{equation}
where $N^\mathrm{acc}$ is the number of speakers for which the peak class of the predicted distribution equals to the ground truth class.

Of course, the most intuitive evaluation metric for SSL should be the mean absolute error (MAE) between the predicted source position and true source position, which can be described as follows:
\begin{equation}\label{eq:mae}
    \mathrm{MAE}(\circ)=\frac{1}{N} \sum_{n=1}^N \mathrm{min}(|\hat p_n - p_n|, 360-|\hat p_n - p_n|)
\end{equation}

\section{Experimental results}\label{sec:res}

\begin{table}[t]
    \centering
    \caption{Experimental results on the dataset of A1, where QE is short for quantization error, the loss function is NLAE, and the backbone network is PNN.}
      \scalebox{0.92}{\begin{tabular}{c|c|ccc}
        \toprule
              & ACC   & \multicolumn{3}{c}{MAE } \\
        \midrule
        Regression & ---     & \multicolumn{3}{c}{0.642} \\
        QE limit & 100.00 & \multicolumn{3}{c}{1.223} \\
        \midrule
              &       & Top-1 & WAD-2 & WAD-3 \\
        One-hot & 98.56 & 1.225 & 0.920 & 0.924 \\
        GLC   & 97.50 & 1.231 & 1.036 & 0.697 \\
        SLD   & 97.44 & 1.237 & 1.437 & 1.144 \\
        \textbf{ULD} & \textbf{98.67} & \textbf{1.224} & \textbf{0.065} & \textbf{0.061} \\
        \bottomrule
        \end{tabular}}
    \label{tab:a1}%
\end{table}%

\begin{table*}[t]
    \centering
    \caption{Results on the dataset of C1, the loss function is NLAE, the encoding method is ULD, and the backbone network is PNN.}
      \scalebox{0.92}{\begin{tabular}{ccccccccccccc}
      \toprule
      \multicolumn{2}{c}{$l$} & Regression & 360   & 180   & 90    & 45    & 20    & 10    & 5     & 3     & 2     & 1 \\
      \midrule
      \multicolumn{2}{c}{ACC } & ---     & 99.21 & 97.18 & 97.64 & 96.30 & 96.39 & 94.68 & 92.27 & 87.59 & 80.51 & 54.77 \\
      \midrule
      \multicolumn{2}{c}{QE limit} & 0     & 89.056 & 45.740 & 22.519 & 11.407 & 5.059 & 2.481 & 1.244 & 0.752 & 0.504 & 0.249 \\
      \midrule
      \multirow{3}[1]{*}{MAE } & Top-1 & \multirow{3}[1]{*}{21.725} & 89.056 & 46.012 & 22.666 & 11.522 & 5.102 & 2.538 & 1.307 & 0.853 & 0.676 & \textbf{0.625} \\
            & WAD-2 &       & 17.562 & 7.336 & 3.226 & 1.847 & 1.028 & 0.793 & 0.649 & 0.557 & 0.547 & \textbf{0.541} \\
            & WAD-3 &       & 17.562 & 10.055 & 3.756 & 1.865 & 1.005 & 0.719 & 0.540 & \textbf{0.476} & 0.486 & 0.561 \\
      \bottomrule
      \end{tabular}}
    \label{tab:c1_reso}%
\end{table*}%

\begin{table}[t]
  \centering
  \caption{Results on the lightly-reverberant dataset C1, where the backbone network is PNN.}
    \scalebox{0.92}{\begin{tabular}{cccccc}
    \toprule
    \multirow{2}[4]{*}{Encoding} & \multirow{2}[4]{*}{Loss} & \multirow{2}[4]{*}{ACC} & \multicolumn{3}{c}{MAE} \\
    \cmidrule{4-6}          &       &       & Top-1 & WAD-2 & WAD-3 \\
    \midrule
    \multirow{5}[2]{*}{One-hot} & CE    & 86.57 & 0.869 & 0.572 & 0.517 \\
          & MSE   & 85.93 & 0.881 & 0.558 & 0.522 \\
          & WD    & 53.61 & 1.681 & 1.716 & 1.701 \\
          \cmidrule{2-6}
          & \textbf{NLAE} & 86.39 & 0.868 & 0.565 & \textbf{0.514} \\
          & MSE (wo) & 84.63 & 0.905 & 0.591 & 0.558 \\
    \midrule
    \multirow{4}[2]{*}{GLC  \cite{he2018deep}} & BCE   & 81.34 & 0.954 & 0.895 & 0.894 \\
          & MSE  \cite{he2018deep} & 80.23 & 0.961 & 0.889 & 0.898 \\
          \cmidrule{2-6}
          & NLAE  & 80.42 & 0.950 & 0.882 & 0.880 \\
          & \textbf{MSE (wo)} & 82.59 & 0.921 & 0.881 & \textbf{0.846} \\
    \midrule
    \multirow{6}[2]{*}{SLD  \cite{subramanian2022deep}} & CE & 80.00 & 0.947 & 0.800 & 0.765 \\
          & BCE   & 81.62 & 0.929 & 0.794 & 0.744 \\
          & MSE   & 83.84 & 0.881 & 0.791 & 0.713 \\
          & WD  \cite{subramanian2022deep} & 55.51 & 1.536 & 1.665 & 1.636 \\
          \cmidrule{2-6}
          & NLAE  & 82.22 & 0.918 & 0.775 & 0.716 \\
          & \textbf{MSE (wo)} & 84.95 & 0.881 & 0.781 & \textbf{0.688} \\
    \midrule
    \multirow{6}[2]{*}{\textbf{ULD}} & CE    & 86.71 & 0.852 & 0.574 & 0.490 \\
          & BCE   & 87.18 & 0.852 & 0.572 & 0.489 \\
          & MSE   & 87.41 & 0.853 & 0.561 & 0.487 \\
          & WD    & 55.69 & 1.567 & 1.629 & 1.606 \\
          \cmidrule{2-6}
          & \textbf{NLAE} & 87.59 & 0.853 & 0.557 & \textbf{0.476} \\
          & MSE (wo) & 87.69 & 0.854 & 0.600 & 0.503 \\
    \bottomrule
    \end{tabular}}
  \label{tab:c1_pnn}
\end{table}%

\begin{table}[t]
  \centering
  \caption{Results on the heavily-reverberant dataset of L1, where the backbone network is PNN.}
    \scalebox{0.92}{\begin{tabular}{cccccc}
    \toprule
    \multirow{2}[4]{*}{Encoding} & \multirow{2}[4]{*}{Loss} & \multirow{2}[4]{*}{ACC} & \multicolumn{3}{c}{MAE} \\
\cmidrule{4-6}          &       &       & Top-1 & WAD-2 & WAD-3 \\
    \midrule
    \multirow{5}[2]{*}{One-hot} & CE    & 68.17 & 3.947 & 3.696 & 3.649 \\
          & MSE   & 68.53 & 3.985 & 3.751 & 3.693 \\
          & WD    & 62.94 & 4.408 & 4.184 & 4.154 \\
          \cmidrule{2-6}
          & \textbf{NLAE} & 65.50 & 3.814 & 3.592 & \textbf{3.539} \\
          & MSE (wo) & 65.78 & 4.040 & 3.874 & 3.787 \\
    \midrule
    \multirow{4}[2]{*}{GLC  \cite{he2018deep}} & BCE   & 62.08 & 3.938 & 3.921 & 3.789 \\
          & MSE  \cite{he2018deep} & 60.94 & 3.924 & 3.883 & 3.761 \\
          \cmidrule{2-6}
          & NLAE  & 62.75 & 3.936 & 3.932 & 3.760 \\
          & \textbf{MSE (wo)} & 59.56 & 3.693 & 3.652 & \textbf{3.552} \\
    \midrule
    \multirow{6}[2]{*}{SLD  \cite{subramanian2022deep}} & \textbf{CE} & 62.67 & 3.670 & 3.694 & \textbf{3.537} \\
          & BCE   & 61.61 & 3.804 & 3.791 & 3.669 \\
          & MSE   & 62.19 & 4.011 & 3.973 & 3.866 \\
          & WD  \cite{subramanian2022deep} & 57.78 & 4.337 & 4.313 & 4.185 \\
          \cmidrule{2-6}
          & NLAE  & 60.67 & 4.763 & 4.725 & 4.630 \\
          & MSE (wo) & 63.78 & 4.005 & 4.006 & 3.871 \\
    \midrule
    \multirow{6}[2]{*}{\textbf{ULD}} & CE    & 67.03 & 3.988 & 3.713 & 3.649 \\
          & BCE   & 69.08 & 3.689 & 3.459 & 3.363 \\
          & MSE   & 68.53 & 4.012 & 3.735 & 3.663 \\
          & WD    & 65.75 & 4.427 & 4.233 & 4.137 \\
          \cmidrule{2-6}
          & NLAE  & 65.44 & 4.632 & 4.385 & 4.295 \\
          & \textbf{MSE (wo)} & 65.39 & 3.481 & 3.179 & \textbf{3.148} \\
    \bottomrule
    \end{tabular}}
  \label{tab:l1_pnn}%
\end{table}%

\begin{table}[t]
  \centering
  \caption{Main results on the real-world data, where the backbone network is PNN.}
  \scalebox{0.8}{
    \begin{tabular}{cccccc}
    \toprule
    \multirow{2}[3]{*}{Test data} & \multirow{2}[3]{*}{Method} & \multirow{2}[3]{*}{ACC} & \multicolumn{3}{c}{MAE} \\
\cmidrule{4-6}          &       &       & Top-1 & WAD-2 & WAD-3 \\
    \midrule
    \multirow{5}[0]{*}{Office} & One-hot + CE & 62.60 & 3.168 & 3.181 & 3.077 \\
          & One-hot + NLAE & 63.17 & 3.071 & 3.087 & 2.978 \\
          & GLC + MSE(wo) & 57.93 & 3.518 & 3.926 & 3.500 \\
          & SLD + MSE(wo) & 60.43 & 3.184 & 3.582 & 3.163 \\
          & \textbf{ULD + MSE(wo)} & 64.12 & 3.008 & 3.116 & \textbf{2.925} \\
    \midrule
    \multirow{5}[1]{*}{Conference} & One-hot + CE & 54.32 & 6.888 & 6.905 & 6.805 \\
          & \textbf{One-hot + NLAE} & 53.84 & 5.192 & 5.268 & \textbf{5.138} \\
          & GLC + MSE(wo) & 49.90 & 5.816 & 6.078 & 5.814 \\
          & SLD + MSE(wo) & 53.71 & 6.052 & 6.498 & 6.052 \\
          & ULD + MSE(wo) & 54.18 & 5.431 & 5.635 & 5.420 \\
    \bottomrule
    \end{tabular}}
  \label{tab:r1_pnn}%
\end{table}%

\begin{table}[t]
  \centering
  \caption{Main results on the single-source data, where the backbone network is Snet, and the loss function for soft labels is MSE (wo).}
  \scalebox{0.85}{
    \begin{tabular}{cccccc}
    \toprule
    \multirow{2}[3]{*}{Test data} & \multirow{2}[3]{*}{Method} & \multirow{2}[3]{*}{ACC} & \multicolumn{3}{c}{MAE} \\
    \cmidrule{4-6}          &       &       & Top-1 & WAD-2 & WAD-3 \\
    \midrule
    \multirow{4}[1]{*}{Simulated data L1} & One-hot & 75.79 & 2.612 & 2.265 & 2.236 \\
    & GLC   & 71.44 & 2.292 & 2.214 & 2.055 \\
    & SLD   & 69.93 & 2.543 & 2.544 & 2.360 \\
    & \textbf{ULD} & 76.88 & 2.176 & 1.782 & \textbf{1.696} \\
    \midrule
    \multirow{4}[0]{*}{Office} & One-hot & 67.48 & 2.285 & 2.197 & 2.169 \\
    & GLC   & 61.11 & 2.474 & 2.766 & 2.413 \\
    & SLD   & 59.88 & 2.555 & 2.894 & 2.521 \\
    & \textbf{ULD} & 68.96 & 2.177 & 2.208 & \textbf{2.145} \\
    \midrule
    \multirow{4}[1]{*}{Conference} & One-hot & 56.57 & 4.616 & 4.488 & 4.483 \\
    & \textbf{GLC} & 53.70 & 4.337 & 4.462 & \textbf{4.257} \\
    & SLD   & 52.97 & 4.519 & 4.696 & 4.465 \\
    & ULD   & 57.18 & 4.583 & 4.476 & 4.456 \\
    \bottomrule
    \end{tabular}%
    }
  \label{tab:l1_snet}%
\end{table}%

\begin{table}[t]
  \centering
  \caption{Main results on the simulated multi-source dataset C2, where the backbone network is PNN-Split.}
    \scalebox{0.92}{\begin{tabular}{cccccc}
    \toprule
    \multirow{2}[4]{*}{Encoding} & \multirow{2}[4]{*}{Loss} & \multirow{2}[4]{*}{ACC} & \multicolumn{3}{c}{MAE} \\
\cmidrule{4-6}          &       &       & Top-1 & WAD-2 & WAD-3 \\
    \midrule
    One-hot & WD    & 68.22 & 7.189 & 6.573 & 6.576 \\
    \textbf{GLC}   & \textbf{MSE (wo)} & 77.09 & 5.855 & 5.391 & \textbf{5.043} \\
    SLD   & MSE (wo) & 71.88 & 5.588 & 5.386 & 5.203 \\
    ULD   & MSE (wo) & 79.33 & 6.089 & 5.291 & 5.114 \\
    \bottomrule
    \end{tabular}}
  \label{tab:c2_pnn}%
\end{table}%

\begin{table}[t]
    \centering
    \caption{Results of the combined encoding method with respect to the parameter $\alpha$ on C2, where the backbone network is PNN-Split, and the loss function is MSE(wo).}
    \scalebox{0.88}{
      \begin{tabular}{cccccc}
      \toprule
      \multirow{2}[4]{*}{Encoding} & \multirow{2}[4]{*}{$\alpha$} & \multirow{2}[4]{*}{ACC} & \multicolumn{3}{c}{MAE} \\
  \cmidrule{4-6}          &       &       & Top-1 & WAD-2 & WAD-3 \\
      \midrule
      \multirow{6}[2]{*}{$\alpha$ULD+$(1-\alpha)$GLC} & 0.0   & 77.09 & 5.855 & 5.391 & 5.043 \\
            & \textbf{0.2} & 78.40 & 5.334 & 4.815 & \textbf{4.465} \\
            & 0.4   & 78.65 & 5.595 & 5.073 & 4.710 \\
            & 0.6   & 77.29 & 5.939 & 5.347 & 5.049 \\
            & 0.8   & 79.22 & 6.082 & 5.423 & 5.167 \\
            & 1.0   & 79.33 & 6.089 & 5.291 & 5.114 \\
      \bottomrule
      \end{tabular}
    }
    \label{tab:c2_joint}
  \end{table}%

\begin{table}[t]
    \centering
    \caption{Results on the highly-reverberant multi-source data, where the backbone network is the SNet-Split. When using one-hot encoding, the training loss function is WD, while for all others it is MSE(wo). The parameter $\alpha$ in ``$\alpha$ULD+$(1-\alpha)$GLC'' is set to 0.2.}
    \scalebox{0.7}{
      \begin{tabular}{cccccc}
      \toprule
      \multirow{2}[4]{*}{Subset} & \multirow{2}[4]{*}{Method} & \multirow{2}[4]{*}{ACC} & \multicolumn{3}{c}{MAE} \\
  \cmidrule{4-6}          &       &       & Top-1 & WAD-2 & WAD-3 \\
      \midrule
      \multirow{5}[2]{*}{Simulated data L2} & One-hot & 62.56 & 6.145 & 5.892 & 5.896 \\
            & GLC   & 70.74 & 4.760 & 4.386 & 4.093 \\
            & SLD   & 70.64 & 4.798 & 4.765 & 4.495 \\
            & ULD   & 73.53 & 5.296 & 4.554 & 4.524 \\
            & \textbf{$\alpha$ULD+$(1-\alpha)$GLC} & 73.01 & 4.446 & 4.008 & \textbf{3.739} \\
      \midrule
      \multirow{5}[2]{*}{Office} & One-hot & 60.47 & 6.422 & 6.309 &6.291 \\
            & GLC   & 64.23 & 6.130 & 6.756 & 6.180 \\
            & SLD   & 63.62 & 6.204 & 7.005 & 6.244 \\
            & ULD   & 62.79 & 6.505 & 6.500 & 6.434 \\
            & \textbf{$\alpha$ULD+$(1-\alpha)$GLC} & 65.22 & 6.107 & 6.620 & \textbf{6.081} \\
      \midrule
      \multirow{5}[2]{*}{Conference} & One-hot & 44.96 & 12.723 & 12.552 & 12.543 \\
            & GLC   & 53.86 & 11.349 & 11.442 & 11.211 \\
            & SLD   & 53.63 & 11.714 & 11.941 & 11.640 \\
            & ULD   & 54.58 & 12.776 & 12.655 & 12.591 \\
            & \textbf{$\alpha$ULD+$(1-\alpha)$GLC} & 53.68 & 10.519 & 10.459 & \textbf{10.229} \\
      \bottomrule
      \end{tabular}%
      }
    \label{tab:l2_exp}%
  \end{table}%

\subsection{Empirical study on breaking quantization error limit}
\subsubsection{Results on anechoic environment}
As the aforementioned, the data for A1 is anechoic and the array itself does not have random angles. The distance between the source and the microphone array is fixed at 1.5m. Overall, this is a very simple dataset that we primarily use to investigate the issue of quantization error. As shown in Table~\ref{tab:a1}, the quantization error limit is 1.223, which confirms our theoretical analysis that the quantization error is approximately $l/4$ (in this case, 1.250). If we use classical Top-1 decoding, then the MAE is bounded by the quantization error limit. However, when we use ULD in conjunction with WAD, the quantization error limit has been significantly broken through. Interestingly, applying WAD to the one-hot encoding actually reduces the MAE and breaks the quantization error limit as well, due to the presence of sidelobes (the loss is not 0).

\subsubsection{Results on reverberant environment}

In the previous section, we have observed that WAD could break the quantization error limit in the anechoic environment, here we study WAD in a reverberant environment. To study the effect of WAD integratively, we tune the the parameter $l$ in a wide range. Specifically, the parameter $l$ determines the azimuth range of a cell (i.e. a class), so as to the number of classes.
When $l$ decreases from 360 to 1, the number of classes naturally increases, which results in an increased model complexity accordingly.

Table~\ref{tab:c1_reso} lists the performance of decoding methods along with the parameter $l$ in the reverberant data C1. From the table, we see that WAD breaks the quantization error limit. Specifically, when $l\geq3$ , WAD yields smaller MAE than the quantization error limit, while the Top-1 decoding yields larger MAE than quantization error. The best performance of WAD appears at $l=3$, and when $l=1$, the MAE of WAD becomes larger than the quantization error limit. This phenomenon is mainly caused by the overfitting problem of the classification model. Specifically, reducing $l$ from 3 to 2 requires increasing the number of classes from 121 to 181, and reducing $l$ from 2 to 1 even requires increasing the number of classes from 181 to 361. As the number of classes increases significantly, the corresponding increase in complexity significantly reduces ACC, which may in turn reduce localization performance.

Another interesting result is that, even when $l$ is as large as 360 (i.e., SSL is formulated as a binary classification problem), the MAE of the proposed WAD is still smaller than the regression model. We believe this is due to that WAD can be seen essentially as performing regression within each class, which integrates the merits of both classification and regression.

\subsection{single-source localization}

\subsubsection{Results on simulated data}

Table~\ref{tab:c1_pnn} lists the performance of various combinations of encoding methods, loss functions, and decoding methods on the lightly-reverberant dataset C1. From the table, it can be seen that ULD is the best label encoding method in this environment; NLAE is the best loss function; and WAD-3 is the best decoding method in almost all cases except with the WD loss function. The reason for the poor performance of the ``ULD+WD loss'' scheme, we believe, is that WD can essentially be viewed as a special form of global regression, therefore inheriting the vulnerability of global regression when used for SSL, such as suffering from confusion similar to that in Eq.~\eqref{eq:mae}. Compared to the most common combination of one-hot encoding with CE and Top-1 decoding, the combination of ULD with NLAE and WAD-3 reduces the MAE by 45.22\%.

Table~\ref{tab:l1_pnn} further lists the performance of comparison methods on the heavily reverberant L1 dataset, where linear arrays have random rotation angles and distances between sound sources and microphone arrays are unconstrained. The table shows the proposed strategy of ULD, MSE(wo), and WAD-3 achieves top performance, 20.24\% above conventional one-hot paradigm. We find NLAE performs poorly in this adverse condition.

\subsubsection{Results on real-world data}

\textbf{First, we note that, due to space constraint of the paper, we only report main experimental results in the following sections, leaving full results in the supplementary material.}

Table~\ref{tab:r1_pnn} lists the performance of comparison models on two real-world datasets, trained on simulated data L1. From the table, we see that: (i) the proposed ULD, MSE(wo), and WAD-3 strategy performs best on the office room dataset; (ii) the proposed NLAE and WAD-3, combined with one-hot encoding, performs best on the conference room dataset, closely followed by the proposed ULD, MSE(wo), and WAD-3 strategy; (iii) as analyzed earlier, the CE loss function underperforms significantly compared to the proposed NLAE and MSE(wo), highlighting the importance of using a loss function with a global receptive field.

\subsubsection{Effect of backbone networks on performance}
In previous experiments, all backbone networks were PNN. In this subsection, we study how backbone networks affect performance. Table \ref{tab:l1_snet} lists comparison method performance using SNet as the backbone network on both simulated data L1 and real-world data. Compared to Table \ref{tab:r1_pnn}, we see the proposed ULD, MSE(wo), and WAD-3 strategy performs best on L1 and the office room, consistent with results using the PNN backbone network. Although the best performance in the conference room appears with proposed MSE(wo) and WAD-3 combined with GLC, proposed ULD, MSE(wo), and WAD-3 performance follows closely, consistent with the PNN backbone results.

\subsection{multi-source localization}

Table~\ref{tab:c2_pnn} lists the performance of the comparison methods on the multi-source data C2. From the table, we see that WAD-3 remains the best decoding method in almost all cases; MSE(wo) still fits the soft labels best, including GLC, SLD, and the proposed ULD. We also observe that GLC is slightly better than the proposed ULD. This phenomenon may be caused by the extreme smoothness of GLC, making it more conducive to model training in challenging scenarios like multi-source localization.

Note that GLC and ULD exhibit distinct advantages---while one emphasizes greater smoothness, the other prioritizes higher precision, suggesting their combined use. We designed a joint training method using weighted loss with training objectives for both ULD and GLC: ``$\alpha$ULD+$(1-\alpha)$GLC'', where $\alpha \in [0,1]$ is a tunable parameter. Table~\ref{tab:c2_joint} lists the performance of the combined encoding method with respect to $\alpha$. From the table, we see the combined encoding method significantly outperforms its components, i.e. GLC and ULD. The best performance of the combined method appears at $\alpha = 0.2$.

Finally, we conducted experiments on a set of highly challenging datasets, L2, and real-world datasets featuring two speakers with significant reverberation and considerable distance from microphones. As shown in Table~\ref{tab:l2_exp}, results are overall consistent with previous experiments, showcasing sustained superior performance of our method. However, given the adverse conditions, ACC is notably low. Classification error predominates over quantization error, limiting performance gains.

\section{Conclusions}    \label{sec:conclusion}
In this paper, we propose a novel output architecture for SSL, incorporating three novel components: (i) ULD as an encoding method, (ii) NLAE and MSE(wo) as training loss functions, and (iii) WAD as a decoding method. Specifically, unlike one-hot encoding, ULD is a one-to-one encoding method, i.e. an unbiased inverse mapping between sound source position and label distribution. NLAE and MSE(wo) are new loss functions integrating benefits of cross-entropy-like and MSE-like functions. They can be viewed as regression-based loss functions applied per class. Unlike Top-1 decoding, WAD considers not only peak class but also sidelobes during decoding. Experimental results on both simulated and real-world data show WAD significantly outperforms quantization error limits, especially with ULD encoding. The proposed NLAE is best for soft labels in simple environments, while MSE(wo) performs best for soft labels in challenging environments. The overall output architecture performs best in most cases.

\appendix
\section*{Appendix A. Proof of Theorem \ref{thm:qe}}
\begin{proof}

  Let $\gamma$ be a non-negative real number, $n$ be the integer part of $\gamma$, and $\mathrm{round}(\gamma)$ be the integer obtained by rounding $\gamma$. Then the expected value of $|\gamma-\mathrm{round}(\gamma)|$ is:
  $$
  \mathbb{E}(|\gamma-\mathrm{round}(\gamma)|) = \int_{n}^{n+1} |x-\mathrm{round}(x)|dx
  $$

  We can split the interval $[n, n+1]$ into two parts: $[n, n+0.5)$ and $[n+0.5, n+1]$. In the interval $[n, n+0.5)$, $\mathrm{round}(x)=n$; in the interval $[n+0.5, n+1]$, $\mathrm{round}(x)=n+1$. Therefore:

  $$
  \begin{aligned}
  &\mathbb{E}(|\gamma-\mathrm{round}(\gamma)|) \\
  &= \int_{n}^{n+0.5} |x-n|dx + \int_{n+0.5}^{n+1} |x-(n+1)|dx \\
  &= \int_{n}^{n+0.5} (x-n)dx + \int_{n+0.5}^{n+1} ((n+1)-x)dx \\
  &= \left[\frac{(x-n)^2}{2}\right]_{x=n}^{x=n+0.5} + \left[(n+1)x-\frac{x^2}{2}\right]_{x=n+0.5}^{x=n+1} \\
  &= \frac{1}{8} + \frac{1}{8} \\
  &= \frac{1}{4}
  \end{aligned}
  $$

  Through the above formula, we can further calculate the mathematical expectation of the quantization error as follows:
  $$
  \begin{aligned}
      \mathbb{E}(qe) &= \mathbb{E}(|p - \sum_{i=0}^I y^{\mathrm{1-hot}}_{i} \times i \times l|) \\
        &= \mathbb{E}(|\gamma \times l - 1 \times \mathrm{round}(\gamma) \times l|) \\
        &= \mathbb{E}(\gamma - \mathrm{round}(\gamma)|) \times l \\
        &= l / 4
  \end{aligned}
  $$
\end{proof}

\section*{Appendix B. Proof of Theorem \ref{thm:uld}}
\begin{proof}
\begin{align*}
    &(1 - \mathrm{deci}(\gamma)) \times \mathrm{int}(\gamma) + \mathrm{deci}(\gamma) \times (\mathrm{int}(\gamma) + 1) \\
      &= \mathrm{int}(\gamma) - \mathrm{deci}(\gamma) \times \mathrm{int}(\gamma) + \mathrm{deci}(\gamma) \times \mathrm{int}(\gamma) + \mathrm{deci}(\gamma) \\
      &= \mathrm{int}(\gamma) + \mathrm{deci}(\gamma) \\
      &= \gamma
\end{align*}
\end{proof}

\section*{Appendix C. Proof of Lemma \ref{lemma:bce}}
\begin{proof}
This is a convex optimization problem. First, we define the Lagrangian function:

$$
L(\hat y, \lambda) = -\sum_{i=1}^{I} \log(1-\hat y_i) + \lambda \left(\sum_{i=1}^{I} \hat y_i - c\right)
$$
where $\lambda$ is the Lagrange multiplier.

Taking the partial derivatives of $L(\hat y, \lambda)$ with respect to $\hat y_i$ and $\lambda$, and setting them to zero, we get:
$$
\frac{\partial L}{\partial \hat y_i} = -\frac{1}{1-\hat y_i} + \lambda = 0
$$
$$
\frac{\partial L}{\partial \lambda} = \sum_{i=1}^{I} \hat y_i - c = 0
$$
From the first equation, we can solve for $\hat y_i = 1 - \frac{1}{\lambda}$. Substituting this into the second equation, we get:
$$
I\left(1-\frac{1}{\lambda}\right) = c
$$
Solving for above, we get $\frac{1}{\lambda} = 1-\frac{c}{I}$. Therefore, $y_i = \frac{c}{I}$.
Substituting the value of $y_i$ into the original expression, we get:

$$
-\sum_{i=1}^{I} \log(1-y_i) = -I\log\left(1-\frac{c}{I}\right)
$$

Therefore, when $y_i = \frac{c}{I}$, $-\sum_{i=1}^{I} \log(1-y_i)$ takes the minimum value of $-I\log\left(1-\frac{c}{I}\right)$.

\end{proof}

\bibliographystyle{elsarticle-num}
\bibliography{mybib}

\begin{thebibliography}{10}
\expandafter\ifx\csname url\endcsname\relax
  \def\url#1{\texttt{#1}}\fi
\expandafter\ifx\csname urlprefix\endcsname\relax\def\urlprefix{URL }\fi
\expandafter\ifx\csname href\endcsname\relax
  \def\href#1#2{#2} \def\path#1{#1}\fi

\bibitem{grumiaux2022survey}
P.-A. Grumiaux, S.~Kiti{\'c}, L.~Girin, A.~Gu{\'e}rin, A survey of sound source localization with deep learning methods, The Journal of the Acoustical Society of America 152~(1) (2022) 107--151.

\bibitem{shimada2021accdoa}
K.~Shimada, Y.~Koyama, N.~Takahashi, S.~Takahashi, Y.~Mitsufuji, Accdoa: Activity-coupled cartesian direction of arrival representation for sound event localization and detection, in: ICASSP 2021-2021 IEEE International Conference on Acoustics, Speech and Signal Processing (ICASSP), IEEE, 2021, pp. 915--919.

\bibitem{shimada2022multi}
K.~Shimada, Y.~Koyama, S.~Takahashi, N.~Takahashi, E.~Tsunoo, Y.~Mitsufuji, Multi-accdoa: Localizing and detecting overlapping sounds from the same class with auxiliary duplicating permutation invariant training, in: ICASSP 2022-2022 IEEE International Conference on Acoustics, Speech and Signal Processing (ICASSP), IEEE, 2022, pp. 316--320.

\bibitem{bai20233d}
J.~Bai, S.~Huang, H.~Yin, Y.~Jia, M.~Wang, J.~Chen, 3d audio signal processing systems for speech enhancement and sound localization and detection, in: ICASSP 2023-2023 IEEE International Conference on Acoustics, Speech and Signal Processing (ICASSP), IEEE, 2023, pp. 1--2.

\bibitem{wang2022localization}
Z.-Q. Wang, D.~Wang, Localization based sequential grouping for continuous speech separation, in: ICASSP 2022-2022 IEEE International Conference on Acoustics, Speech and Signal Processing (ICASSP), IEEE, 2022, pp. 281--285.

\bibitem{taherian2022lbt}
H.~Taherian, K.~Tan, D.~Wang, Multi-channel talker-independent speaker separation through location-based training, IEEE/ACM Transactions on Audio, Speech, and Language Processing 30 (2022) 2791--2800.

\bibitem{subramanian2021directional}
A.~S. Subramanian, C.~Weng, S.~Watanabe, M.~Yu, Y.~Xu, S.-X. Zhang, D.~Yu, Directional asr: A new paradigm for e2e multi-speaker speech recognition with source localization, in: ICASSP 2021-2021 IEEE International Conference on Acoustics, Speech and Signal Processing (ICASSP), IEEE, 2021, pp. 8433--8437.

\bibitem{subramanian2022deep}
A.~S. Subramanian, C.~Weng, S.~Watanabe, M.~Yu, D.~Yu, Deep learning based multi-source localization with source splitting and its effectiveness in multi-talker speech recognition, Computer Speech \& Language 75 (2022) 101360.

\bibitem{zheng2021real}
S.~Zheng, W.~Huang, X.~Wang, H.~Suo, J.~Feng, Z.~Yan, A real-time speaker diarization system based on spatial spectrum, in: ICASSP 2021-2021 IEEE International Conference on Acoustics, Speech and Signal Processing (ICASSP), IEEE, 2021, pp. 7208--7212.

\bibitem{taherian2023multi}
H.~Taherian, D.~Wang, Multi-channel conversational speaker separation via neural diarization, arXiv preprint arXiv:2311.08630 (2023).

\bibitem{knapp1976generalized}
C.~Knapp, G.~Carter, The generalized correlation method for estimation of time delay, IEEE transactions on acoustics, speech, and signal processing 24~(4) (1976) 320--327.

\bibitem{schmidt1986multiple}
R.~Schmidt, Multiple emitter location and signal parameter estimation, IEEE transactions on antennas and propagation 34~(3) (1986) 276--280.

\bibitem{dibiase2000high}
J.~H. DiBiase, A high-accuracy, low-latency technique for talker localization in reverberant environments using microphone arrays, Brown University, 2000.

\bibitem{chakrabarty2019multi}
S.~Chakrabarty, E.~A. Habets, Multi-speaker doa estimation using deep convolutional networks trained with noise signals, IEEE Journal of Selected Topics in Signal Processing 13~(1) (2019) 8--21.

\bibitem{vesperini2016neural}
F.~Vesperini, P.~Vecchiotti, E.~Principi, S.~Squartini, F.~Piazza, A neural network based algorithm for speaker localization in a multi-room environment, in: 2016 IEEE 26th International Workshop on Machine Learning for Signal Processing (MLSP), IEEE, 2016, pp. 1--6.

\bibitem{vecchiotti2019detection}
P.~Vecchiotti, G.~Pepe, E.~Principi, S.~Squartini, Detection of activity and position of speakers by using deep neural networks and acoustic data augmentation, Expert Systems with Applications 134 (2019) 53--65.

\bibitem{vera2018towards}
J.~M. Vera-Diaz, D.~Pizarro, J.~Macias-Guarasa, Towards end-to-end acoustic localization using deep learning: From audio signals to source position coordinates, Sensors 18~(10) (2018) 3418.

\bibitem{adavanne2018sound}
S.~Adavanne, A.~Politis, J.~Nikunen, T.~Virtanen, Sound event localization and detection of overlapping sources using convolutional recurrent neural networks, IEEE Journal of Selected Topics in Signal Processing 13~(1) (2018) 34--48.

\bibitem{diaz2020robust}
D.~Diaz-Guerra, A.~Miguel, J.~R. Beltran, Robust sound source tracking using srp-phat and 3d convolutional neural networks, IEEE/ACM Transactions on Audio, Speech, and Language Processing 29 (2020) 300--311.

\bibitem{diaz2022direction}
D.~Diaz-Guerra, A.~Miguel, J.~R. Beltran, Direction of arrival estimation of sound sources using icosahedral cnns, IEEE/ACM Transactions on Audio, Speech, and Language Processing 31 (2022) 313--321.

\bibitem{xiao2015learning}
X.~Xiao, S.~Zhao, X.~Zhong, D.~L. Jones, E.~S. Chng, H.~Li, A learning-based approach to direction of arrival estimation in noisy and reverberant environments, in: 2015 IEEE International Conference on Acoustics, Speech and Signal Processing (ICASSP), IEEE, 2015, pp. 2814--2818.

\bibitem{he2019adaptation}
W.~He, P.~Motlicek, J.-M. Odobez, Adaptation of multiple sound source localization neural networks with weak supervision and domain-adversarial training, in: ICASSP 2019-2019 IEEE International Conference on Acoustics, Speech and Signal Processing (ICASSP), IEEE, 2019, pp. 770--774.

\bibitem{nguyen2020robust}
T.~N.~T. Nguyen, W.-S. Gan, R.~Ranjan, D.~L. Jones, Robust source counting and doa estimation using spatial pseudo-spectrum and convolutional neural network, IEEE/ACM Transactions on Audio, Speech, and Language Processing 28 (2020) 2626--2637.

\bibitem{fu2022iterative}
Y.~Fu, M.~Ge, H.~Yin, X.~Qian, L.~Wang, G.~Zhang, J.~Dang, {Iterative Sound Source Localization for Unknown Number of Sources}, in: Proc. Interspeech 2022, 2022, pp. 896--900.
\newblock \href {https://doi.org/10.21437/Interspeech.2022-10525} {\path{doi:10.21437/Interspeech.2022-10525}}.

\bibitem{tang19_interspeech}
Z.~Tang, J.~D. Kanu, K.~Hogan, D.~Manocha, {Regression and Classification for Direction-of-Arrival Estimation with Convolutional Recurrent Neural Networks}, in: Proc. Interspeech 2019, 2019, pp. 654--658.
\newblock \href {https://doi.org/10.21437/Interspeech.2019-1111} {\path{doi:10.21437/Interspeech.2019-1111}}.

\bibitem{perotin2019regression}
L.~Perotin, A.~D{\'e}fossez, E.~Vincent, R.~Serizel, A.~Gu{\'e}rin, Regression versus classification for neural network based audio source localization, in: 2019 IEEE Workshop on Applications of Signal Processing to Audio and Acoustics (WASPAA), IEEE, 2019, pp. 343--347.

\bibitem{feng2023soft}
L.~Feng, Y.~Gong, X.-L. Zhang, Soft label coding for end-to-end sound source localization with ad-hoc microphone arrays, in: ICASSP 2023-2023 IEEE International Conference on Acoustics, Speech and Signal Processing (ICASSP), IEEE, 2023, pp. 1--5.

\bibitem{he2018deep}
W.~He, P.~Motlicek, J.-M. Odobez, Deep neural networks for multiple speaker detection and localization, in: 2018 IEEE International Conference on Robotics and Automation (ICRA), IEEE, 2018, pp. 74--79.

\bibitem{rubner2000earth}
Y.~Rubner, C.~Tomasi, L.~J. Guibas, The earth mover's distance as a metric for image retrieval, International journal of computer vision 40~(2) (2000) 99.

\bibitem{panayotov2015librispeech}
V.~Panayotov, G.~Chen, D.~Povey, S.~Khudanpur, Librispeech: an asr corpus based on public domain audio books, in: 2015 IEEE international conference on acoustics, speech and signal processing (ICASSP), IEEE, 2015, pp. 5206--5210.

\bibitem{scheibler2018pyroomacoustics}
R.~Scheibler, E.~Bezzam, I.~Dokmani{\'c}, Pyroomacoustics: A python package for audio room simulation and array processing algorithms, in: 2018 IEEE international conference on acoustics, speech and signal processing (ICASSP), IEEE, 2018, pp. 351--355.

\bibitem{tan2021speech}
X.~Tan, X.-L. Zhang, Speech enhancement aided end-to-end multi-task learning for voice activity detection, in: ICASSP 2021-2021 IEEE International Conference on Acoustics, Speech and Signal Processing (ICASSP), IEEE, 2021, pp. 6823--6827.

\bibitem{liu2022deep}
S.~Liu, Y.~Gong, X.-L. Zhang, Deep learning based two-dimensional speaker localization with large ad-hoc microphone arrays, arXiv preprint arXiv:2210.10265 (2022).

\bibitem{levina2001earth}
E.~Levina, P.~Bickel, The earth mover's distance is the mallows distance: Some insights from statistics, in: Proceedings Eighth IEEE International Conference on Computer Vision. ICCV 2001, Vol.~2, IEEE, 2001, pp. 251--256.

\bibitem{he2021neural}
W.~He, P.~Motlicek, J.-M. Odobez, Neural network adaptation and data augmentation for multi-speaker direction-of-arrival estimation, IEEE/ACM Transactions on Audio, Speech, and Language Processing 29 (2021) 1303--1317.

\end{thebibliography}

\end{document}